\documentclass[sigconf]{acmart}
\renewcommand\footnotetextcopyrightpermission[1]{} 
\usepackage{booktabs} 
\usepackage{tabularx}
\usepackage{algorithm}
\usepackage[noend]{algpseudocode}
\usepackage{algorithmicx}
\usepackage{xspace}
\usepackage{caption}
\usepackage[skip=1mm]{subcaption}
\usepackage{enumitem}
\usepackage{xcolor,colortbl}
\usepackage{amssymb}
\usepackage{multirow}
\usepackage{hyperref}
\usepackage{etoolbox}

\newcommand{\method}{\textsc{PMV}\xspace}
\newcommand{\methodV}{\textsc{PMV}$_{\text{vertical}}$\xspace}
\newcommand{\methodH}{\textsc{PMV}$_{\text{horizontal}}$\xspace}
\newcommand{\methodS}{\textsc{PMV}$_{\text{selective}}$\xspace}
\newcommand{\methodHB}{\textsc{PMV}$_{\text{hybrid}}$\xspace}

\newcommand{\combTwo}{\texttt{combine2}\xspace}

\newcommand{\combAll}{\texttt{combineAll}\xspace}

\newcommand{\assign}{\texttt{assign}\xspace}

\newcommand{\hide}[1]{}
\newcommand{\rulesep}{\unskip\ \vrule\ }

\newcommand{\zerodisplayskips}{%
  \setlength{\abovedisplayskip}{3pt}%
  \setlength{\belowdisplayskip}{3pt}%
  \setlength{\abovedisplayshortskip}{0pt}%
  \setlength{\belowdisplayshortskip}{0pt}}
\appto{\normalsize}{\zerodisplayskips}
\appto{\small}{\zerodisplayskips}
\appto{\footnotesize}{\zerodisplayskips}

\algdef{SNe}[FUNCTION]{Function}{EndFunction}%
[2]{\algorithmicfunction\ #1\ifthenelse{\equal{#2}{}}{}{(#2)}}

\newcommand{\ForEach}[1]{\For{\textbf{each} #1}}

\algdef{SNe}[FORPARALLEL]{ForParallel}{EndForParallel}%
[1]{\textbf{for each} #1 \textbf{do in parallel}}

\algdef{SNe}[CLASS]{Class}{EndClass}%
[1]{\textbf{Class} #1}{}

\setcopyright{none}

\begin{document}
\title{PMV: Pre-partitioned Generalized Matrix-Vector Multiplication \\ for Scalable Graph Mining}


\author{Chiwan Park}
\affiliation{%
  \institution{Seoul National University}
}
\email{chiwanpark@snu.ac.kr}

\author{Ha-Myung Park}
\affiliation{%
  \institution{KAIST}
}
\email{hamyung.park@kaist.ac.kr}

\author{Minji Yoon}
\affiliation{%
	\institution{Seoul National University}
}
\email{riin0716@gmail.com}

\author{U Kang}
\affiliation{
  \institution{Seoul National University}
}
\email{ukang@snu.ac.kr}

\renewcommand{\shortauthors}{C. Park et al.}

\begin{abstract}
How can we analyze enormous networks including the Web and social networks which have hundreds of billions of nodes and edges?
Network analyses have been conducted by various graph mining methods including shortest path computation, PageRank, connected component computation, {random walk with restart}, etc.
These graph mining methods can be expressed as generalized matrix-vector multiplication which consists of few operations inspired by typical matrix-vector multiplication.
Recently, several graph processing systems based on matrix-vector multiplication or their own primitives have been proposed to deal with large graphs; however, they all have failed on Web-scale graphs due to insufficient memory space or the lack of consideration for I/O costs.

In this paper, we propose PMV (Pre-partitioned generalized Matrix-Vector multiplication), a scalable distributed graph mining method based on generalized matrix-vector multiplication on distributed systems.
PMV significantly decreases the communication cost, which is the main bottleneck of distributed systems, by partitioning the input graph in advance and judiciously applying execution strategies based on the density of the pre-partitioned sub-matrices.
Experiments show that PMV succeeds in processing up to $16\times$ larger graphs than existing distributed memory-based graph mining methods, and requires $9\times$ less time than previous disk-based graph mining methods by reducing I/O costs significantly. 
\end{abstract}

\maketitle

\section{Introduction}
\label{sec:introduction}
\vspace{-1mm}
How can we analyze enormous networks including the Web and social networks which have hundreds of billions of nodes and edges?
Various graph mining algorithms {including shortest path computation~\cite{Dijkstra:ShortestPath, DBLP:journals/cacm/Floyd62a}, PageRank~\cite{DBLP:journals/cn/BrinP98}, connected component computation~\cite{DBLP:journals/cacm/HopcroftT73, DBLP:journals/jacm/EvenS81}, and random walk with restart~\cite{DBLP:conf/eccv/GradyF04}}, have been developed for network analyses and many of them are expressed in generalized matrix-vector multiplication form~\cite{DBLP:conf/icdm/KangTF09}.
As graph sizes increase exponentially, many efforts have been devoted to find scalable graph processing methods which could perform large-scale matrix-vector multiplication efficiently on distributed systems.

\begin{figure}[t]
	\centering
	\includegraphics[width=85mm]{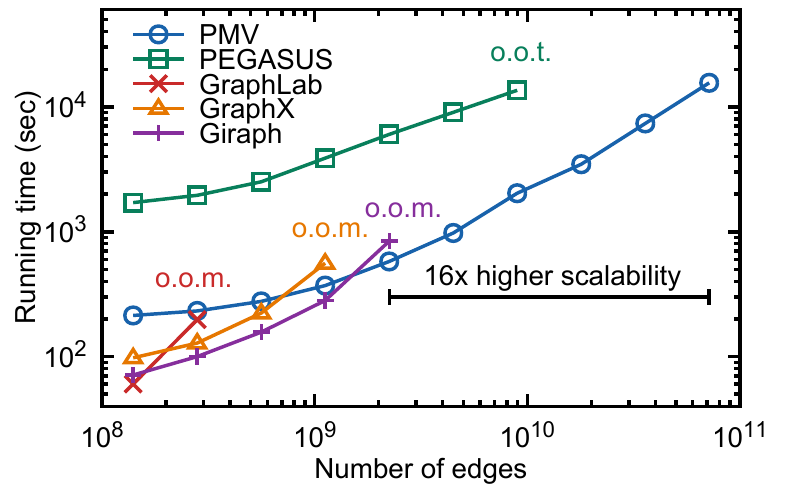}
	\vspace{-2mm}
	\caption{The running time on subgraphs of ClueWeb12. o.o.m.: out of memory. o.o.t.: out of time (>5h). Our proposed method \method is the only framework that succeeds in processing the full ClueWeb12 graph.}
	\label{fig:data-scale-runtime}
	\vspace{-2mm}
\end{figure}

Recently, several graph processing systems have been proposed to perform such computations in billion-scale graphs; they are divided into single-machine systems, distributed-memory, and MapReduce-based systems.
However, they all have limited scalability.
I/O efficient single-machine systems including GraphChi~\cite{conf/osdi/KyrolaBG12} cannot process a graph exceeding the external-memory space of a single machine.
Similarly, distributed-memory systems like GraphLab~\cite{powergraph} cannot process a graph that does not fit into the distributed-memory.
On the other hand, MapReduce-based systems~\cite{DBLP:conf/icdm/KangTF09, sgc, pte, netray, DBLP:conf/icdm/JeonJK15}, which use a distributed-external-memory like GFS~\cite{gfs} or HDFS~\cite{hdfs}, can handle much larger graphs than single-machine or distributed-memory systems do.
However, the MapReduce-based systems succeed only in non-iterative graph mining tasks such as triangle counting~\cite{pte, DBLP:conf/cikm/ParkSKP14} and graph visualization~\cite{DBLP:conf/icdm/JeonJK15, netray}.
They have limited scalability for iterative tasks like PageRank because they need to read and shuffle the entire input graph in every iteration.
In MapReduce~\cite{DBLP:journals/cacm/DeanG08}, shuffling massive data is the main performance bottleneck as it requires heavy disk and network I/Os, which seriously limit the scalability and the fault tolerance.
Thus, it is desirable to shrink the amount of shuffled data to process  matrix-vector multiplication in distributed systems.


In this paper, we propose \method (Pre-partitioned generalized Matrix-Vector multiplication), a new scalable graph mining algorithm performing large-scale generalized matrix-vector multiplication in distributed systems.
\method succeeds in processing billion-scale graphs which all other state-of-the-art distributed systems fail to process, by significantly reducing the shuffled data size, and the costs of network and disk I/Os.
\method partitions the matrix of input graph once, and reuses the partitioned matrices for all iterations.
Moreover, \method carefully assigns the partitioned matrix blocks to each worker to minimize the I/O cost. \method is a general framework that can be implemented in any distributed framework; we implement \method on Hadoop and Spark, the two most widely used distributed computing frameworks.
Our main contributions are the following:

\begin{itemize}[topsep=0.5mm]
\item{
  \textbf{Algorithm.}
  We propose \method, a new scalable graph mining algorithm for performing generalized matrix-vector multiplication in distributed systems.
  \method is designed to reduce the amount of shuffled data by partitioning the input matrix before iterative computation. Moreover, \method splits the partitioned matrix blocks into two regions and applies different placement strategies on them to minimize the I/O cost.
}
\item{
  \textbf{Cost analysis.}
  We give a theoretical analysis of the I/O costs of the block placement strategies which are the criteria of block placement selection.
  We prove the efficiency of \method by giving theoretical analyses of the performance.
}
\item{
  \textbf{Experiment.}
  We empirically evaluate \method using both large real-world and synthetic networks.
  We emphasize that only our system succeeds in processing the Clueweb12 graph which has 6 billion vertices and 71 billion edges.
  Also, \method shows up to $9\times$ faster performance than previous MapReduce-based methods do (see Figure~\ref{fig:data-scale-runtime}).
}
\end{itemize}

The rest of the paper is organized as follows.
In Section~\ref{sec:background_and_related_works}, we review existing large-scale graph processing systems and introduce GIM-V primitive for graph mining tasks.
In Section~\ref{sec:proposed_method}, we describe the proposed algorithm \method in detail along with its theoretical analysis.
After showing experimental results in Section~\ref{sec:experiments}, we conclude in Section~\ref{sec:conclusion}.
The symbols frequently used in this paper are summarized in Table~\ref{tbl:symbols}.

\begin{table}[t]
\small
\setlength{\tabcolsep}{1.5pt}
\caption{\textbf{Table of symbols.}}
\vspace{-5mm}
\begin{center}
\begin{tabularx}{0.49\textwidth}{cX}
\toprule
\textbf{Symbol} & \textbf{Description}                                                           \\ \midrule
$v$             & Vector, or set of vertices                                                     \\
$\theta$        & Degree threshold to divide sparse and dense sub-matrices \\
$out(p)$				& a set of out-neighbors of a vertex $p$ \\
$b$             & Number of vector blocks or vertex partitions                                   \\
$\psi$          & Vertex partitioning function: $v \rightarrow \{1, ..., b\}$                    \\
$v_i$           & $i$-th element of $v$                                                           \\
$v^{(i)}$       & Set of vector elements $(p, v_p)$ where $\psi(p) = i$                          \\
$v_s^{(i)}$     & Set of vector elements $(p, v_p) \in v^{(i)}$ where $|out(p)| < \theta$ \\
$v_d^{(i)}$     & Set of vector elements $(p, v_p) \in v^{(i)}$ where $|out(p)| \ge \theta$ \\
$|v|$           & Size of vector $v$, or of vertices in a graph                                  \\
$M$             & Matrix, or set of edges                                                        \\
$m_{i,j}$       & $(i, j)$-th element of $M$                                                      \\
$M^{(i,j)}$     & Set of matrix elements $(p, q, m_{p,q})$ where $\psi(p) = i$ and $\psi(q) = j$ \\
$M_s^{(i,j)}$    & Set of matrix elements $(p, q, m_{p,q}) \in M^{(i,j)}$ where $|out(q)| < \theta$ \\
$M_d^{(i,j)}$   & Set of matrix elements $(p, q, m_{p,q}) \in M^{(i,j)}$ where $|out(q)| \ge \theta$ \\
$|M|$           & Number of non-zero elements in $M$ (= number of edges in a graph)                        \\
$\otimes$       & User-defined matrix-vector multiplication                                      \\
\bottomrule
\end{tabularx}
\end{center}
\label{tbl:symbols}
\vspace{-2mm}
\end{table}

\section{Background and Related Works}
\label{sec:background_and_related_works}
\vspace{-1mm}

In this section, we first review representative graph processing systems and show their limitations on scalability (Section~\ref{sec:other-systems}).
Then, we outline MapReduce and Spark to highlight the importance of decreasing the amount of shuffled data in improving their performances (Sections~\ref{sec:mapreduce-and-spark}).
After that, we review the GIM-V model for graph algorithms (Section~\ref{sec:gimv}).


\subsection{Large-scale Graph Processing Systems}
\label{sec:other-systems}
\vspace{-1mm}


Large-scale graph processing systems can be classified into three groups: I/O efficient single-machine systems, distributed-memory systems, and MapReduce-based systems.

I/O efficient graph mining systems~\cite{conf/osdi/KyrolaBG12, conf/kdd/HanLPL0KY13, conf/bigdataconf/LinKSCLK14, conf/pkdd/GualdronCRCKK16} handle large graphs with external-memory (i.e., disk) and optimize disk I/O costs to achieve higher performance.
Some single-machine systems~\cite{DBLP:journals/pvldb/YangPS11, DBLP:conf/ppopp/SeoKK15, DBLP:conf/usenix/MaYCXD17} use accelerators like GPUs to improve performance.
However, all of these systems have limited scalability as they use only a single machine.

A typical approach to handle large-scale graphs is using multiple machines.
Recently, several graph processing systems using distributed-memory have been proposed:
Pregel~\cite{pregel}, GraphLab-PowerGraph~\cite{DBLP:journals/pvldb/LowGKBGH12, powergraph}, Trinity~\cite{trinity}, GraphX~\cite{graphx}, GraphFrames~\cite{DBLP:conf/grades/DaveJLXGZ16}, GPS~\cite{gps}, Presto~\cite{DBLP:conf/eurosys/VenkataramanBRAS13}, Pregel+~\cite{DBLP:conf/www/YanCLN15} and PowerLyra~\cite{DBLP:conf/eurosys/ChenSCC15}.
Even though these distributed-memory systems achieve faster performance and higher scalability than single machine systems do, they cannot process graphs that do not fit into the distributed-memory.
Pregelix~\cite{DBLP:journals/pvldb/BuBJCC14} succeeds in processing graphs whose size exceeds the distributed-memory space by exploiting out-of-core support of Hyracks~\cite{DBLP:conf/icde/BorkarCGOV11}, a general data processing engine. However, Pregelix uses only a single placement strategy which is similar to \methodV, one of our basic proposed methods.

MapReduce-based systems increase the processable graph size as MapReduce is a disk-based distributed system.
PEGASUS~\cite{DBLP:conf/icdm/KangTF09, DBLP:books/sp/14/KangF14} is a MapReduce-based graph mining library based on a generalized matrix-vector multiplication.
SGC~\cite{sgc} is another MapReduce-based system exploiting two join operations, namely \textsf{NE join} and \textsf{EN join}.
The MapReduce-based systems, however, still have limited scalability because they need to shuffle the input matrix and vector repeatedly. 
UNICORN~\cite{Lee201556} avoids massive data shuffling by exploiting HBase, a distributed database system on Hadoop, but it reaches another performance bottleneck, intensive random accesses to HBase.

In the next section, we highlight the importance of reducing the amount of shuffled data in MapReduce and Spark.

\subsection{MapReduce and Spark}
\label{sec:mapreduce-and-spark}
\vspace{-1mm}

MapReduce is a programming model to process large data by parallel and distributed computation.
Thanks to its ease of use, fault tolerance, and high scalability, MapReduce has been applied to various graph mining tasks including computation of radius~\cite{hadi}, triangle~\cite{pte}, visualization~\cite{netray}, etc.
MapReduce transforms an input set of key-value pairs into another output set of key-value pairs through three steps: map, shuffle, and reduce.
Each input key-value pair is transformed into a set of key-value pairs (map-step), and all the output pairs from the map-step are grouped by key (shuffle-step), then, each group of pairs is processed independently of other groups. Finally, an output set of key-value pairs is emitted (reduce-step).
The performance of a MapReduce algorithm depends mainly on the amount of shuffled data which are sorted by key requiring massive network and disk I/Os~\cite{herodotou2011hadoop}.
In each map worker, the output pairs from the map-step are stored in $R$ independent regions on disk according to the key where $R$ is the number of reduce workers (collect and spill).
Each map worker outputs key-value pairs into $R$ independent regions on local disks according to the key where $R$ is the number of reduce workers. 
The pairs stored in $R$ regions are shuffled to corresponding reduce workers periodically.
As a reduce worker has received all the pairs from the map workers, the reduce worker conducts external-sort to group the key-value pairs according to the key. 
in order to group the pairs by key (reduce). 
The performance of a MapReduce algorithm depends mainly on the amount of shuffled data since they require massive network and disk I/Os~\cite{herodotou2011hadoop}.
Requiring such heavy disk and network I/Os, a large amount of shuffled data significantly increases the running time and decreases the stability of the system.
Requiring such heavy disk and network I/Os significantly increases the running time and decreases the scalability of the system.
Thus, it is important to shrink the amount of shuffled data as much as possible to increase the performance.

Spark~\cite{DBLP:conf/hotcloud/ZahariaCFSS10} is a general data processing engine with an abstraction of data collection called Resilient Distributed Datasets (RDDs) \cite{DBLP:conf/nsdi/ZahariaCDDMMFSS12}.
Each RDD consists of multiple partitions distributed across the machines of a cluster.
Each partition has data objects and can be manipulated through operations like \texttt{map} and \texttt{reduce}.
Unlike Hadoop, a widely used open-source implementation of MapReduce, RDD partitions are cached in memory or on disks of each worker in the cluster.
Due to the in-memory caching, Spark shows a good performance for iterative computation~\cite{DBLP:journals/pvldb/ShiQMJWRO15, DBLP:conf/cikm/LeeKYLK16} which is necessary for graph mining and machine learning tasks.
However, Spark still requires disk I/O \cite{DBLP:conf/nsdi/OusterhoutRRSC15} since its typical operations with shuffling including \texttt{join} and \texttt{groupBy} operations need to access disks for external-sort.
Therefore, the effort to reduce intermediate data to be shuffled is still valuable in Spark.

\begin{table}[!tp]
	\small
	\setlength{\tabcolsep}{1.5pt}
	\caption{\textbf{Graph Algorithms on GIM-V}}
	\vspace{-4mm}
	\begin{center}
		\begin{tabularx}{0.49\textwidth}{cX}
			\toprule
			\textbf{Algorithm} & \textbf{GIM-V Functions} \\
			\midrule
			\multirow{3}{*}{PageRank} & \combTwo$(m_{i,j}, v_j)=m_{i,j} \times v_j$ \\
			& \combAll$(\{x_{i,1}, \cdots, x_{i,n}\})=\sum_{i=1}^n x_i$ \\
			& \assign$(v_i, r_i)=0.15+0.85\times r_i$ \\
			\midrule
      \multirow{5}{*}{\shortstack[c]{Random Walk\\with Restart}} & \combTwo$(m_{i,j}, v_j)=m_{i,j} \times v_j$ \\
      & \combAll$(\{x_{i,1}, \cdots, x_{i,n}\})=\sum_{i=1}^n x_i$ \\
      & \assign$(v_i, r_i)=\begin{cases}
          0.15 + 0.85 \times r_i & \text{if } i = \text{source vertex} \\
          0.85 \times r_i & \text{otherwise}
        \end{cases}$\\
      \midrule
			\multirow{3}{*}{\shortstack[c]{Single Source\\Shortest Path}} & \combTwo$(m_{i,j}, v_j)=m_{i,j} + v_j$ \\
			& \combAll$(\{x_{i,1}, \cdots, x_{i,n}\})=\min(\{x_{i,1}, \cdots, x_{i,n} \})$ \\
			& \assign$(v_i, r_i)=\min(v_i, r_i)$ \\
			\midrule
			\multirow{3}{*}{\shortstack[c]{Connected\\Component}} & \combTwo$(m_{i,j}, v_j)=v_j$ \\
			& \combAll$(\{x_{i,1}, \cdots, x_{i,n}\})=\min(\{x_{i,1}, \cdots, x_{i,n} \})$ \\
			& \assign$(v_i, r_i)=\min(v_i, r_i)$ \\
			\bottomrule
		\end{tabularx}
	\end{center}
	\label{tbl:gimv-algorithms}
	\vspace{-3mm}
\end{table}

\subsection{GIM-V for Graph Algorithms}
\label{sec:gimv}
\vspace{-1mm}
\begin{figure}[!t]
  \begin{subfigure}[t]{0.49\linewidth}
    \includegraphics[width=41mm]{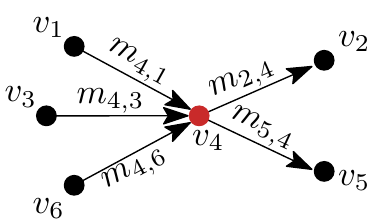}
    \caption{A graph with 6 vertices}
    \label{fig:graph_example_graph}
  \end{subfigure}
  \begin{subfigure}[t]{0.49\linewidth}
    \includegraphics[width=41mm, page=3]{fig/graph_matrix_vector}
    \caption{Matrix-vector representation}
    \label{fig:graph_example_matrix}
  \end{subfigure}
  \vspace{-3mm}
	\caption{\textbf{An example graph with 6 vertices and 9 edges and its matrix-vector representation.} In GIM-V, the vertex 4 receives 3 messages from incoming neighbors (1, 3, and 6) and sends 2 messages to outgoing neighbors (2 and 5). This process can be represented by matrix-vector multiplication.}
	\label{fig:graphmatrixvector}
\end{figure}

{
Several optimized algorithms have been proposed for specific graph mining tasks such as shortest path computation~\cite{DBLP:journals/pvldb/KalavriSL16, DBLP:conf/sigmod/NavlakhaRS08, de2012stochastic}, connected component computation~\cite{DBLP:conf/ipps/PatwaryRM12}, and random walk with restart~\cite{DBLP:conf/sigmod/ShinJSK15, DBLP:journals/tods/JungSSK16, DBLP:conf/icdm/JungJSK16, DBLP:conf/sigmod/JungPSK17}.
GIM-V (Generalized Iterative Matrix-Vector Multiplication)~\cite{DBLP:conf/icdm/KangTF09}, a widely-used graph mining primitive, unifies such graph algorithms by representing them in the form of matrix-vector multiplication.
%
%
}%
For GIM-V representation, a user needs to describe only three operations for a graph algorithm: \combTwo, \combAll, and \assign.

Consider a matrix $M$ of size $n \times n$, and a vector $v$ of size $n$, where $m_{i,j}$ is the $(i,j)$-th element of $M$, and $v_i$ is the $i$-th element of $v$ for $i,j \in \{ 1, \cdots, n \}$.
Then, the operations play the following roles:
\begin{sloppypar}
	\begin{itemize}
		\item{ $\combTwo(m_{i,j}, v_{j})$: return the combined value $x_{i,j}$ from a matrix element $m_{i,j}$ and a vector element $v_{j}$. }
		\item{ $\combAll(\{x_{i,1}, \cdots, x_{i,n} \})$: reduce the input values to a single value $r_i$.}
		\item{ $\assign(v_{i}, r_{i})$: compute the new $i$-th vector element $v'_i$ for the next iteration from the current $i$-th vector element $v_{i}$ and the reduced value $r_{i}$, and check the convergence. }
	\end{itemize}
\end{sloppypar}


Let $M \otimes v$ be a user-defined generalized matrix-vector multiplication between the matrix $M$ and the vector $v$.
The new $i$-th vector element $v'_i$ of the result vector $v'$ of $M \otimes v$ is then:
\begin{align*}
  v'_i &= \assign(v_i, \combAll( \\
  & \{ x_{i, j} | x_{i,j} = \combTwo(m_{i,j}, v_j), j \in \{ 1, \cdots, n \} \}))
\end{align*}

GIM-V can be considered as a process of passing messages from each vertex to its outgoing neighbors on a graph where $m_{i,j}$ corresponds to an edge from vertex $j$ to vertex $i$. 
In Figure~\ref{fig:graphmatrixvector}, vertex $4$ receives messages $\{ x_{4,1}, x_{4,3}, x_{4,6} \}$ from incoming neighbors $1$, $3$, and $6$, where $x_{4, j} = \combTwo(m_{4,j}, v_j)$ for $j \in \{ 1, 3, 6 \}$.
From the received messages, GIM-V calculates a new value $r_4 = \combAll(\{ x_{4,1}, x_{4,3}, x_{4,6} \})$ for the vertex $4$,
and then, updates $v_4$ with a new value $v'_4 = \assign(v_4, r_4)$.
The updated value $v'_4$ is passed to the outgoing neighbors $2$ and $5$ in the next iteration.

With GIM-V, a user can easily describe various graph algorithms.
Table~\ref{tbl:gimv-algorithms} shows implementations of PageRank, {random walk with restart,} single source shortest path, and connected component on GIM-V, respectively.
Note that only few lines of codes are required for the implementations.

\section{Proposed Method}
\label{sec:proposed_method}
\vspace{-1mm}
\begin{figure*}[!t]
  \centering
  \includegraphics[width=175mm, page=4]{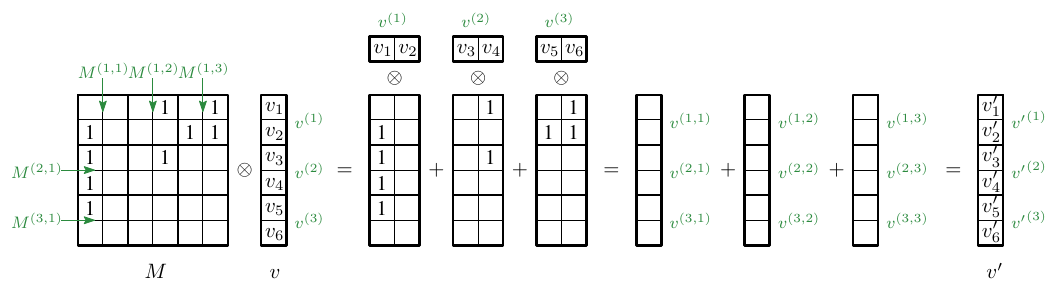}
  \vspace{-5mm}
  \caption{\textbf{The user-defined generalized matrix-vector multiplication $M \otimes v$ performed on $3 \times 3$ sub-matrices.} $M^{(i,j)}$ is $(i, j)$-th sub-matrix and $v^{(i)}$ is $i$-th sub-vector. $v^{(i,j)}$ is the result vector of sub-multiplication $M^{(i,j)} \otimes v^{(j)}$. The $i$-th sub-vector ${v'}^{(i)}$ of the result vector $v'$ is calculated by combining $v^{(i,j)}$ for $j \in \{1, \cdots, b\}$ with the \combAll ($+$) operation.}
  \label{fig:block_multiplication}
  \vspace{-5mm}
\end{figure*}

In this section, we propose \method, a scalable algorithm to efficiently perform the GIM-V on distributed systems.
\method greatly increases the scalability by the following ideas:
\begin{enumerate}
	\item{ \textit{Pre-partitioning} significantly shrinks the amount of shuffled data. \method shuffles $O(|M|)$ data only once at the beginning while the previous MapReduce algorithms shuffle $O(|M| + |v|)$ data in each iteration (Section~\ref{subsec:method-pmv}). }
	\item{ \textit{Considering the density of the pre-partitioned matrices} enables \method to minimize the I/O cost by applying the two multiplication methods: vertical placement and horizontal placement (Sections~\ref{subsec:method-horizontal}-\ref{subsec:pmv-hybrid}).
}
\end{enumerate}

We first describe the pre-partitioning method in Section~\ref{subsec:method-pmv}.
Once the graph is partitioned, the multiplication method can be classified as \methodH and \methodV depending on which partitions are processed together on the same machine.
We describe the two basic methods in Sections~\ref{subsec:method-horizontal} and \ref{subsec:method-vertical}.
In Section~\ref{subsec:pmv-selective}, we analyze the I/O cost of \methodH and \methodV, and propose a na\"ive method, namely \methodS, that selects one of the two basic methods according to the density of the input graph.
After that, we propose \methodHB, our desired method, that uses the two basic methods simultaneously in Section~\ref{subsec:pmv-hybrid}.
Finally, in Section~\ref{subsec:implementation}, we describe how to implement \method on two popular distributed frameworks, Hadoop and Spark, to show that \method is general enough to be implemented on any computing frameworks.
\subsection{\method: Pre-partitioned Generalized Matrix-Vector Multiplication}
\label{subsec:method-pmv}
\vspace{-1mm}

How can we efficiently perform GIM-V on distributed systems?
The key idea of \method is based on the observation that the input matrix $M$ never changes and is reused in each iteration, while the vector $v$ varies.
\method first divides the vector $v$ into several sub-vectors and partitions the matrix $M$ into corresponding sub-matrices which will be multiplied with each sub-vector respectively.
Then only sub-vectors are shuffled to the corresponding sub-matrices in the iteration phase, thus avoiding shuffling the entire matrix in every iteration unlike existing MapReduce-based systems which shuffle the entire matrix.
Note that, even though some distributed-memory systems also do not shuffle the matrix by retaining both the matrix and the vector in main memory of each worker redundantly, they fail when the matrix and the vector do not fit into the memory while \method is insensitive to the memory size.
\method consists of two steps: the pre-partitioning and the iterative multiplication.

\vspace{-2mm}
\subsubsection{Pre-partitioning.}

%
%

\method first initializes the input vector $v$ properly based on the graph algorithm used. For example, $v$ is set to $1/|v|$ in PageRank.
Then, \method partitions the matrix $M$ into $b \times b$ sub-matrices $M^{(i,j)} = \{ m_{p,q} \in M \ | \ \psi(p) = i, \psi(q) = j \}$ for $i,j \in \{ 1, \cdots, b \}$ where $\psi$ is a vertex partitioning function.
Likewise, the vector $v$ is also divided into $b$ sub-vectors $v^{(i)} = \{ v_{p} \in v \ | \ \psi(p) = i \}$ for $i \in \{ 1, \cdots, b \}$.
We consider the number of workers and the size of vector to determine the number $b$ of blocks.
$b$ is set to the number $W$ of workers to maximize the parallelism if $|v|/\mathcal{M} < W$, otherwise $b$ is set to $O(|v|/\mathcal{M})$ to fit a sub-vector into the main memory of size $\mathcal{M}$.
Note that this proper setting for $b$ makes \method insensitive to the memory size.
In Figure \ref{fig:graph_example_matrix}, the partitioning function $\psi$ divides the set of vertices $\{1,2,3,4,5,6\}$ into $b = 3$ subsets $\{1,2\}, \{3,4\}, \{5,6\}$.
Accordingly, the matrix and the vector are divided into $3 \times 3$ sub-matrices, and $3$ sub-vectors, respectively; sub-matrices and sub-vectors are depicted with boxes with bold border lines.
\vspace{-2mm}

\subsubsection{Iterative Multiplication.}
\method divides the entire problem $M \otimes v$ into $b^2$ subproblems and solves them in parallel.
Subproblem $\langle i,j \rangle$ is to calculate $v^{(i,j)} = M^{(i,j)} \otimes v^{(j)}$ for each pair $(i, j) \in \{1, \cdots,b\}^2$. Then, $i$-th sub-vector $v'^{(i)}$ is calculated by combining $v^{(i,j)}$ for all $j \in \{ 1, \cdots, b \}$.
Figure~\ref{fig:block_multiplication} illustrates how the entire problem is divided into several subproblems in \method.
A subproblem requires $O(|v|/b)$ of the memory size:
a subproblem should retain a sub-vector $v^{(i)}$, whose expected size is $O(|v|/b)$, in the main memory of a worker. 
The sub-matrix $M^{(i,j)}$ is cached in the main memory or external memory of a worker:
each worker reads the sub-matrix once from distributed storage and stores it locally.

Meanwhile, each worker solves multiple subproblems.
The way of distributing subproblems to workers affects the amount of I/Os.
Then, how should we assign the subproblems to workers to minimize the I/O cost?
In the following subsections, we introduce multiple \method methods to answer the question.
We focus on the I/O cost of handling only vectors because all the methods require the same I/O cost $O(|M|)$ to read the matrix by the local caching of sub-matrices.

\subsection{\methodH: Horizontal Matrix Placement}
\label{subsec:method-horizontal}
\vspace{-1mm}

\begin{algorithm}[t]
  \caption{Iterative Multiplication (\methodH)}
  \label{alg:pmv-horizontal}
  \begin{algorithmic}[1]
  	\small
    \Require a set $\{ (M^{(i,:)}, v) \ | \ i \in \{ 1, \cdots, b\} \}$ of matrix-vector pairs
    \Ensure a result vector $v' = \{v'^{(i)} \ | \ i \in \{ 1, \cdots, b\} \}$

    \Repeat
    \ForParallel{$(M^{(i,:)} , v)$}
    \State initialize $v'^{(i)}$
    \ForEach{$j \in \{ 1, \cdots, b \}$}
    \State $v^{(i,j)} \leftarrow \combAll_b(\combTwo_b(M^{(i,j)}, v^{(j)}))$
    \State $v'^{(i)} \leftarrow \combAll_b(v^{(i,j)} \cup v'^{(i)})$
    \EndFor
    \State $v'^{(i)} \leftarrow \assign_b(v^{(i)}, v'^{(i)})$
    \State store $v'^{(i)}$ to $v^{(i)}$ in distributed storage
    \EndForParallel
    \Until{convergence}
    \State \Return $v' = \bigcup_{i \in \{1, \cdots, b\}}{ v^{(i)} }$
  \end{algorithmic}
\end{algorithm}

\methodH uses horizontal matrix placement illustrated in Figure~\ref{fig:horizontal-placement} so that each worker solves subproblems which share the same output sub-vector.
As a result, \methodH does not need to shuffle any intermediate vector while the input vector is copied multiple times as described in Algorithm \ref{alg:pmv-horizontal}.
Each worker directly computes $v'^{(i)}$ from $M^{(i,:)} = \{ M^{(i,j)} \ | \ j \in \{ 1, \cdots, b \} \}$ and $v$ (lines 2-10).
For $j \in \{1, \cdots, b\}$, a worker computes intermediate vectors $v^{(i,j)}$ by combining $M^{(i,j)}$ and $v^{(j)}$, and reduces them into $v'^{(i)}$ immediately without any access to the distributed storage (lines 4-7). 
Note that $\combAll_b$ and $\combTwo_b$ are block operations for \combAll and \combTwo, respectively; $\combTwo_b(M^{(i,j)}, v^{(j)})$ applies $\combTwo(m_{p,q}, v_q)$ for all $m_{p,q} \in M^{(i,j)}$ and $v_q \in v^{(j)}$, and $\combAll_b(X^{(i,j)})$ reduces each row values in $X^{(i,j)}$ into a single value by applying the \combAll operation.
After that, each worker applies the $\assign_b$ operation where $\assign_b(v^{(j)}, x^{(j)})$ applies $\assign(v_p, x_p)$ for all vertices in $\{p \ | \ v_p \in v^{(j)} \}$ and stores the result to the distributed storage (lines 8-9).
\methodH repeats this task until convergence.

\subsection{\methodV: Vertical Matrix Placement}
\label{subsec:method-vertical}
\vspace{-1mm}

\begin{algorithm}[t]
	\caption{Iterative Multiplication (\methodV)}
	\label{alg:pmv-vertical}
	\begin{algorithmic}[1]
		\small
		\Require a set $\{ (M^{(:,j)}, v^{(j)}) \ | \ j \in \{ 1, \cdots, b\} \}$ of matrix-vector pairs
		\Ensure a result vector $v' = \{v'^{(j)} \ | \ j \in \{ 1, \cdots, b\} \}$
		
		\Repeat
		\ForParallel{$(M^{(:,j)}  , v^{(j)})$}
		\ForEach{sub-matrix $M^{(i,j)} \in M^{(:,j)} $}
  		\State $v^{(i,j)} \leftarrow \combAll_b(\combTwo_b(M^{(i,j)}, v^{(j)}))$
  		\State store $v^{(i,j)}$ to distributed storage
		\EndFor
		\State \textbf{Barrier}
		\State load $v^{(j,i)}$ for $i \in \{1, \cdots, b\} \setminus \{j\}$
		\State $v'^{(j)} \leftarrow \assign_b(v^{(j)}, \combAll_b( \bigcup_{i \in \{1, \cdots, b\}}{ v^{(j,i)} } ))$
		\State store $v'^{(j)}$ to $v^{(j)}$ in distributed storage
		\EndForParallel
		\Until{convergence}
    \State\Return $v' = \bigcup_{j \in \{1, \cdots, b\}}v^{(j)}$
	\end{algorithmic}
\end{algorithm}

\methodV uses vertical matrix placement illustrated in Figure~\ref{fig:vertical-placement} to solve the subproblems that share the same input sub-vector in the same worker.
By doing so, \methodV reads each sub-vector only once in each worker.
As described in Algorithm \ref{alg:pmv-vertical}, \methodV computes $v^{(:,j)} = \{ v^{(i,j)} \ | \ i \in \{1, \cdots, b \} \}$ for each $j \in \{ 1, \cdots, b \}$ in parallel (lines 2-11).
Given $j \in \{ 1, \cdots, b \}$, a worker first loads $v^{(j)}$ into the main memory; then, it computes $v^{(i,j)}$ by sequentially reading $M^{(i,j)}$ for each $i \in \{ 1, \cdots, b \}$ and stores $v^{(i,j)}$ into the distributed storage (lines 3-6).
The worker of $j$ is responsible for combining all intermediate data $v^{(j, i)}$ for $i \in \{1, \cdots, b \}$ stored in the distributed storage into the final value $v^{(j)}$.
After waiting for all the other workers to finish the sub-multiplication using a barrier (line 7), the worker of $j$ loads $v^{(j,i)}$ for $i \in \{1, \cdots, b \}$ from the distributed storage (line 8).
Then, the worker calculates $v'^{(j)}$ which replaces $v^{(j)}$ in the distributed storage (lines 9-10).
Note that, the vectors $v^{(j,i)}$ do not need to be loaded all at once because the \combAll operation is commutative and associative.
\methodV repeats this task until convergence.

\subsection{\methodS: Selecting Best Method between \methodH and \methodV}
\label{subsec:pmv-selective}
\vspace{-1mm}

\begin{algorithm}[t]
  \caption{Iterative Multiplication (\methodS)}
  \label{alg:pmv-selective}
  \begin{algorithmic}[1]
    \Require a set $\{M^{(i,j)}\ |\ (i,j) \in \{1, \cdots, b\}^2\}$ of matrix blocks, a set $\{v^{(i)}\ |\ i \in \{1, \cdots, b\}\}$ of vector blocks
    \Ensure a result vector $v' = \{v'^{(j)} \ | \ j \in \{ 1, \cdots, b\} \}$

    \If{$\left(1-|M|/|v|^2\right)^{|v|/b} < 0.5$}
    \State $v'\leftarrow$ \methodH$(\{ (M^{(i,:)}, v) \ | \ i \in \{ 1, \cdots, b\} \})$
    \Else
    \State $v'\leftarrow$ \methodV$(\{ (M^{(:,j)}, v^{(j)}) \ | \ j \in \{ 1, \cdots, b\} \})$
    \EndIf
  \end{algorithmic}
\end{algorithm}

Given a graph, how can we decide the best multiplication method between \methodH and \methodV? 
In distributed graph systems, a major bottleneck is not computational cost, but expensive I/O cost.
\methodS compares the expected I/O costs of two basic methods and selects the one having the minimum expected I/O cost.
The expected I/O costs of \methodH and \methodV are derived in Lemmas~\ref{lem:io-cost-horizontal} and \ref{lem:io-cost-vertical}, respectively.

\begin{lemma}[I/O Cost of Horizontal Placement]
  \vspace{-2mm}
  \label{lem:io-cost-horizontal}
  \methodH has an expected I/O cost $C_{h}$ per iteration:
  \begin{equation}
  \label{eqn:io-cost-with-horizontal-placement}
  \mathbb{E}\left[C_{h}\right] = (b+1)|v|
  \end{equation}
  where $|v|$ is the size of vector $v$ and $b$ is the number of vector blocks.
  \vspace{-2mm}
\end{lemma}

\begin{proof}
  With the horizontal placement, each worker should load all vector blocks from the distributed storage since \combTwo function in each worker is computed with all sub-vectors.
  This causes $b|v|$ I/O cost.
  Also the result vector $|v|$ should be written to the distributed storage.
  Thus, the total I/O cost is $(b+1)|v|$.
  Note that \methodH requires no communication between workers.
\end{proof}

\begin{lemma}[I/O Cost of Vertical Placement]
  \vspace{-2mm}
  \label{lem:io-cost-vertical}
  \methodV has an expected I/O cost $C_{v}$ per iteration:
  \begin{equation}
    \label{eqn:io-cost-with-vertical-placement}
    \mathbb{E}\left[C_{v}\right] = 2|v|\left(1+(b-1)\left(1-\left(1-{|M|}/{|v|^2}\right)^{|v|/b}\right)\right)
  \end{equation}
  where $|v|$ is the size of vector $v$, $|M|$ is the number of non-zero elements in the matrix $M$, and $b$ is the number of vector blocks.
  \vspace{-2mm}
\end{lemma}

\begin{proof}
  The expected I/O cost of \methodV is the sum of 1) the cost to read the vector from the previous iteration, 2) the cost to transfer the sub-multiplication results between workers using distributed storage, and 3) the cost to write the result vector to the distributed storage.
  To transfer one of the sub-multiplication results, \methodV requires $2\left|v^{(i,j)}\right|$ of I/O costs: one is for writing the results to distributed storage, and the other is for reading them from the distributed storage.
  Therefore,
  \begin{align}
  \begin{split}
    \label{eqn:expected-cost-vertical-abstract}
      \mathbb{E}\left[C_{v}\right] &= 2|v| + \sum_{i \ne j} \mathbb{E}\left[2\left|v^{(i,j)}\right|\right] = 2|v| + 2b(b - 1)\mathbb{E}\left[\left|v^{(i,j)}\right|\right]
  \end{split}
  \end{align}
  where $v^{(i,j)}$ is the result vector of sub-multiplication $M^{(i,j)} \otimes v^{(j)}$.
  For each vertex $u \in v^{(i,j)}$, let $X_u$ denote an event that $u$-th element of $v^{(i,j)}$ has a non-zero value.
  Then,
  \begin{align*}
    \begin{split}
      \mathbf{P}(X_u) &= 1 - \mathbf{P}(u \text{ has no in-edges in } M^{(i, j)}) = 1 - \left(1 - \frac{|M|}{|v|^2}\right)^{|v|/b}
    \end{split}
  \end{align*}
  by assuming that every matrix block has the same number of edges (non-zeros).
  The expected size of the sub-multiplication result is:
  \begin{align}
  \begin{split}
    \label{eqn:size-of-sub-multiplication-result}
      \mathbb{E}\left[|v^{(i,j)}|\right] &= \sum_{u \in v^{(i)}}\mathbf{P}(X_u) = \frac{|v|}{b}\left(1 - \left(1 - \frac{|M|}{|v|^2}\right)^{|v|/b}\right)
  \end{split}
  \end{align}

  \noindent Combining \eqref{eqn:expected-cost-vertical-abstract} and \eqref{eqn:size-of-sub-multiplication-result}, we obtain the claimed I/O cost.
\end{proof}
\vspace{-2mm}

Lemmas \ref{lem:io-cost-horizontal} and \ref{lem:io-cost-vertical} state that the cost depends on the density of the matrix and the number of vector blocks.
Comparing \eqref{eqn:io-cost-with-horizontal-placement} and \eqref{eqn:io-cost-with-vertical-placement}, the condition to prefer horizontal placement over vertical placement is given by \eqref{eqn:placement-decision-equation}.

\vspace{-4mm}
\begin{align}
  \begin{split}
  \label{eqn:placement-decision-equation}
  \mathbb{E}\left[C_{h}\right] < & \mathbb{E}\left[C_{v}\right] \Leftrightarrow \left(1-{|M|}/{|v|^2}\right)^{|v|/b} < 0.5
  \end{split}
\end{align}

\begin{figure*}[!t]
  \begin{subfigure}[t]{35mm}
    \centering
    \includegraphics[height=28mm, page=10]{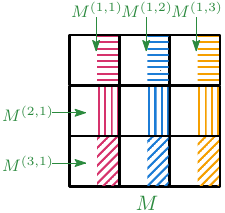}
    \caption{Pre-partitioned Matrix}
    \label{fig:prepartitioned-matrix}
  \end{subfigure}
  \hfil
  \rulesep
  \hfil
  \begin{subfigure}[t]{43mm}
    \centering
    \includegraphics[height=28mm, page=11]{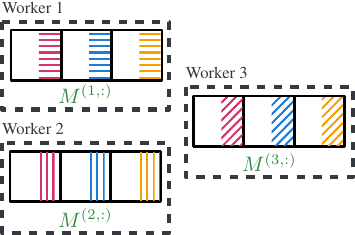}
    \caption{Horizontal Placement}
    \label{fig:horizontal-placement}
  \end{subfigure}
  \hfil
  \rulesep
  \hfil
  \begin{subfigure}[t]{32mm}
    \centering
    \includegraphics[height=28mm, page=12]{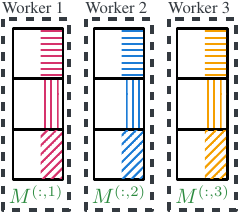}
    \caption{Vertical Placement}
    \label{fig:vertical-placement}
  \end{subfigure}
  \hfil
  \rulesep
  \hfil
  \begin{subfigure}[t]{57mm}
    \centering
    \includegraphics[height=28mm, page=13]{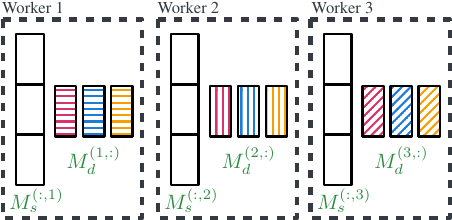}
    \caption{Hybrid Placement}
    \label{fig:hybrid-placement}
  \end{subfigure}
  \caption{An example of matrix placement methods on $3\times3$ sub-matrices with 3 workers. In each matrix block $M^{(i,j)}$, the striped region and the white region represent the dense region $M_d^{(i,j)}$ and the sparse region $M_s^{(i,j)}$, respectively. The horizontal placement groups the matrix blocks $M^{(i,:)}$ which share rows to a worker $i$ while the vertical placement groups the matrix blocks $M^{(:,j)}$ which share columns to a worker $j$. The hybrid placement groups the sparse regions $M_s^{(i, :)}$ which share rows (same stripe pattern) and the dense regions $M_d^{(:, i)}$ which share columns (same color) to a worker $i$.}
  \label{fig:matrix-placements}
  \vspace{-4mm}
\end{figure*}


For sparse matrices, the I/O cost of \methodV is lower than that of \methodH.
On the other hand, for dense matrices, \methodH has smaller I/O cost than that of \methodV.
As described in Algorithm~\ref{alg:pmv-selective}, \methodS first evaluates the condition \eqref{eqn:placement-decision-equation} and selects the best method based on the result.
Thus, the performance of \methodS is better than or at least equal to those of \methodV or \methodH.
Our experiment (see Section~\ref{sec:experiment-performance-density}) shows the effectiveness of each method according to the matrix density.
\subsection{\methodHB: Using \methodH and \methodV Together}
\label{subsec:pmv-hybrid}
\vspace{-1mm}


\methodHB improves \methodS to further reduce I/O costs by using \methodH and \methodV together.
The main idea is based on the fact that \methodV is appropriate for a sparse matrix while \methodH is appropriate for a dense matrix, as we discussed in Section~\ref{subsec:pmv-selective}.
We also observe that density of a matrix block varies across different sub-areas of the block. In other words, some areas of each matrix block are relatively dense with many high-degree vertices while the other areas are sparse.
Using these observations,
\methodHB divides each vector block $v^{(i)}$ into a sparse region $v_s^{(i)}$ with vertices whose out-degrees are smaller than a threshold $\theta$ and a dense region $v_d^{(i)}$ with vertices whose out-degrees are larger than or equal to the threshold.
Likewise, each matrix block $M^{(i,j)}$ is also divided into a sparse region $M_s^{(i,j)}$ where each source vertex is in $v_s^{(j)}$ and a dense region $M_d^{(i,j)}$ where each source vertex is in $v_d^{(j)}$.
Then, \methodHB executes \methodH for the dense area and \methodV for the sparse area.
Figure~\ref{fig:hybrid-placement} illustrates \methodHB on $3\times3$ matrix blocks with 3 workers.

Algorithm~\ref{alg:pmv-hybrid} describes \methodHB. \methodHB performs an additional pre-processing step after the pre-partitioning step to split each matrix block into the dense and sparse regions (lines 1-2).
Then, each worker first multiplies all assigned sparse matrix-vector pairs $(M_s^{(:,j)}, v_s^{(j)})$ by applying \methodV (lines 5-11).
After that, the dense matrix-vector pairs $(M_d^{(j,:)}, v_d^{(:)})$ are multiplied using \methodH and added to the results of the sparse regions (lines 12-16).
Finally, each worker splits the result vector into two regions again for next iteration (lines 17-19).
\methodHB repeats this task until convergence like \methodH and \methodV do.

\begin{algorithm}[t]
  \caption{Iterative Multiplication (\methodHB)}
  \label{alg:pmv-hybrid}
  \begin{algorithmic}[1]
  	\small
    \Require a set $\{M^{(i,j)}\ |\ (i,j) \in \{1, \cdots, b\}^2\}$ of matrix blocks, a set $\{v^{(i)}\ |\ i \in \{1, \cdots, b\}\}$ of vector blocks.
    \Ensure a result vector $v' = \{v'^{(i)}\ |\ i \in \{1, \cdots, b\} \}$

    \State split $v^{(i)}$ into $v_s^{(i)}$ and $v_d^{(i)}$ for $i \in \{1, \cdots b\}$
    \State split $M^{(i,j)}$ into $M_s^{(i,j)}$ and $M_d^{(i,j)}$ for $(i,j) \in \{1, \cdots, b\}^2$

    \Repeat
    \ForParallel{$(M_s^{(:,j)}, v_s^{(j)}, M_d^{(j,:)}, v_d^{(:)})$}
    \ForEach{sub-matrix $M_s^{(i,j)} \in M_s^{(:,j)}$}
    \State $v_s^{(i,j)} \gets \combAll_b(\combTwo_b(M_s^{(i,j)}, v_s^{(j)}))$
    \State store $v_s^{(i,j)}$ to distributed storage
    \EndFor
    \State \textbf{Barrier}
    \State load $v_s^{(j,i)}$ for $i \in \{1, \cdots, b\} \setminus \{j\}$
    \State $v'^{(j)} \gets \combAll_b(\bigcup_{i \in \{1, \cdots, b\}}{ v_s^{(j,i)} } ))$
    \ForEach{$i \in \{1, \cdots, b\}$}
    \State $v_d^{(j,i)} \gets \combAll_b(\combTwo_b(M_d^{(j,i)}, v_d^{(i)}))$
    \State $v'^{(j)} \gets \combAll_b(v_d^{(j,i)}, v'^{(j)})$
    \EndFor
    \State $v'^{(j)} \gets \assign_b(v_s^{(j)} \cup v_d^{(j)}, v'^{(j)})$
    \State split $v'^{(j)}$ into ${v'}_s^{(j)}$ and ${v'}_d^{(j)}$
    \State store ${v'}_s^{(j)}$ to $v_s^{(j)}$ in distributed storage
    \State store ${v'}_d^{(j)}$ to $v_d^{(j)}$ in distributed storage
    \EndForParallel
    \Until{convergence}
    \State\Return $v' = \bigcup_{i \in \{1, \cdots, b\}}v_s^{(i)} \cup v_d^{(i)}$
  \end{algorithmic}
\end{algorithm}

The threshold $\theta$ to split the sparse and dense regions affects the performance and the I/O cost of \methodHB.
If we set $\theta = 0$, \methodHB is the same as \methodH because there is no vertex in the sparse regions.
On the other hand, if we set $\theta = \infty$, \methodHB is the same as \methodV because there is no vertex in the dense regions.
To find the threshold which minimizes the I/O cost, we compute the expected I/O cost of \methodHB varying $\theta$ by Lemma~\ref{lem:cost-hybrid}, and choose the one with the minimum I/O cost.

\begin{lemma}[I/O Cost of \methodHB]
  \vspace{-2mm}
  \label{lem:cost-hybrid}
  \methodHB has an expected I/O cost $C_{hb}$ per iteration:
  \begin{align}\small
    \vspace{-3mm}
    \begin{split}
      \mathbb{E}\left[C_{hb}\right] =&\ |v|\left(\mathbf{P}_{out}(\theta) + b(1 - \mathbf{P}_{out}(\theta)) + 1\right) \\
      &+ 2|v|(b-1)\sum_{d=0}^{|v|}\left(1 - \left(1 - \frac{1}{b}\cdot \mathbf{P}_{out}(\theta)\right)^{d}\right)\cdot p_{in}(d)
    \end{split}
    \vspace{-3mm}
  \end{align}
  \noindent where $|v|$ is the size of vector $v$, $b$ is the number of vector blocks, $\mathbf{P}_{out}(\theta)$ is the ratio of vertices whose out-degree is less than $\theta$, and $p_{in}(d)$ is the ratio of vertices whose in-degree is $d$.
\end{lemma}
\begin{proof}
  The expected I/O cost of \methodHB is the sum of 1) the cost to read the sparse regions of each vector block, 2) the cost to transfer the sub-multiplication results, 3) the cost to read the dense regions of each vector block, and 4) the cost to write the result vector.
  Like \methodV, \methodHB requires $2\left|v_s^{(i,j)}\right|$ of I/O costs to transfer one of the sub-multiplication results by writing the results to distributed storage and reading them from the distributed storage.
  Therefore,
  \begin{align}\small
    \vspace{-3mm}
    \label{eqn:cost-hybrid-form}
    \begin{split}
      \mathbb{E}\left[C_{hb}\right] =&\ |v|\cdot \mathbf{P}_{out}(\theta) + \sum_{i \ne j}\mathbb{E}\left[2\left|v_s^{(i,j)}\right|\right] \\&+ b|v|\cdot (1 - \mathbf{P}_{out}(\theta)) + |v|
    \end{split}
    \vspace{-3mm}
  \end{align}
  \noindent where $v_s^{(i,j)}$ is the result vector of sub-multiplication between $M_s^{(i,j)}$ and $v_s^{(j)}$. For each vertex $u \in v_s^{(i,j)}$, let $X_u$ denote an event that $u$-th element of $v_s^{(i,j)}$ has a non-zero value. Then,
  \begin{align*}\small
    \vspace{-3mm}
    \begin{split}
      \mathbf{P}(X_u) &= 1 - \mathbf{P}(u \text{ has no in-edges in } M_s^{(i,j)})
      = 1 - \left(1 - \frac{\mathbf{P}_{out}(\theta)}{b} \right)^{|in(u)|}
    \end{split}
    \vspace{-3mm}
  \end{align*}
  \noindent where $in(u)$ is a set of in-neighbors of vertex $u$.
  Considering $|in(u)|$ as a random variable following the in-degree distribution $p_{in}(d)$,
  \begin{align}\small
    \vspace{-2mm}
    \label{eqn:cost-hybrid-subresults}
    \begin{split}
      \mathbb{E}\left[\left|v_s^{(i,j)}\right|\right] &= \sum_{u \in v^{(i)}} \mathbf{P}(X_u) = \frac{|v|}{b}\cdot \mathbb{E}\left[\mathbf{P}(X_u)\right] \\
      &= \frac{|v|}{b}\sum_{d=0}^{|v|}\left(1 - \left(1 - \frac{1}{b}\cdot \mathbf{P}_{out}(\theta)\right)^{d}\right)\cdot p_{in}(d)
    \end{split}
  \end{align}
  Combining \eqref{eqn:cost-hybrid-form} and \eqref{eqn:cost-hybrid-subresults}, we obtain the claimed I/O cost.
\end{proof}
\vspace{-2mm}

Note that the in-degree distribution $p_{in}(d)$ and the cumulative out-degree distribution $\mathbf{P}_{out}(\theta)$ are approximated well using  power-law degree distributions for real world graphs. Although the exact cost of \methodHB in Lemma~\ref{lem:cost-hybrid} includes data-dependent terms and thus is not directly comparable to those of other \method methods,
in Section~\ref{sec:experiment-performance-density}
we experimentally show that
\methodHB achieves higher performance and smaller amount of I/O than other \method methods.

\subsection{Implementation}
\label{subsec:implementation}
\vspace{-1mm}

In this section, we discuss practical issues to implement \method on distributed systems.
We only discuss the issues related to \methodHB because \methodH and \methodV are special cases of \methodHB, as we discussed in Section~\ref{subsec:pmv-hybrid}.
We focus on famous distributed processing frameworks, Hadoop and Spark.
Note that \method can be implemented on any distributed processing frameworks.

\vspace{-2mm}
\subsubsection{\method on Hadoop.}
The pre-partitioning is implemented in a single MapReduce job.
The implementation places the matrix blocks within the same column into a single machine;
each matrix element $m_{p,q} \in M^{(i,j)}$ moves to $j$-th reducer during map and shuffle steps; after that, each reducer groups matrix elements into matrix blocks, and divides each matrix block into two regions (sparse and dense) by the given threshold $\theta$.
The iterative multiplication is implemented in a single Map-only job.
Each mapper solves the assigned subproblems one by one; for each subproblem, a mapper reads the corresponding sub-matrix and the sub-vector from HDFS.
The mapper first computes the sub-multiplication $M_s^{(i,j)} \otimes v_s^{(j)}$ of sparse regions, and waits for all the other mappers to finish the sub-multiplication using a barrier.
The result vector $v_s^{(i,j)}$ of a subproblem is sent to the $i$-th mapper via HDFS to be merged to $v'^{(i)}$.
After that, the sub-multiplications $M_d^{(i,:)} \otimes v_d^{(:)}$ of dense regions are computed by the $i$-th mapper.
The result vector $v_d^{(i,j)}$ is directly merged to $v'^{(i)}$ in the main memory.
After the result vector $v'$ is computed, each mapper splits the result vector into the sparse and dense regions.
Then, the next iteration starts with the new vector $v'$ by the same Map-only job until convergence.

\vspace{-2mm}
\subsubsection{\method on Spark.}
The pre-partitioning is implemented by a \texttt{partitionBy} and two \texttt{mapPartitions} operations of typical Resilient Distributed Dataset (RDD) API.
The \texttt{partitionBy} operation uses a custom partitioner to partition the matrix blocks.
The \texttt{mapPartitions} operations output four RDDs, sparseMatRDD, denseMatRDD, sparseVecRDD, and denseVecRDD which contain sparse and dense regions of matrix blocks, and sparse and dense regions of vector blocks, respectively.
Each iteration of matrix-vector multiplication is implemented by five RDD operations.
For the sparse regions, the multiplication comprises the following operations: (1) \texttt{join} operation on the sparseMatRDD and the sparseVecRDD to combine vector blocks and matrix blocks, (2) \texttt{mapPartitions} operation to create the partial vector blocks, and (3) \texttt{reduceByKey} operation on the partial vector blocks.
In the case of the dense regions, each iteration of the multiplication comprises the following operations: (1) \texttt{flatMap} operation on the denseVecRDD to copy the vector blocks, (2) \texttt{join} operation on the denseMatRDD and the copied denseVecRDD, and (3) \texttt{mapPartitions} operation to create the updated vecRDD.
After both multiplications for the sparse and dense regions, (4) \texttt{join} operation is used to combine the results of multiplications in sparse regions and dense regions.
Finally, (5) \texttt{mapPartitions} splits the combined results into sparseVecRDD and denseVecRDD again.
We ensure the colocation of relevant matrix blocks and vector blocks by using a custom partitioner.
Therefore, each worker runs the \texttt{join} operation combining the sparse matrices and the sparse vectors without network I/Os.
The \texttt{join} operation for the dense regions requires network I/Os but only the dense vectors, whose sizes are relatively small in \methodHB, are transferred.

\section{Experiments}
\label{sec:experiments}
\vspace{-1mm}
\begin{figure}[!t]
  \centering
  \hfil
  \includegraphics[width=90mm]{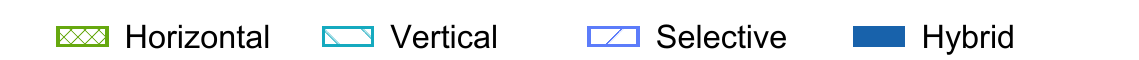}
  \begin{subfigure}[t]{0.49\linewidth}
    \includegraphics[width=40mm]{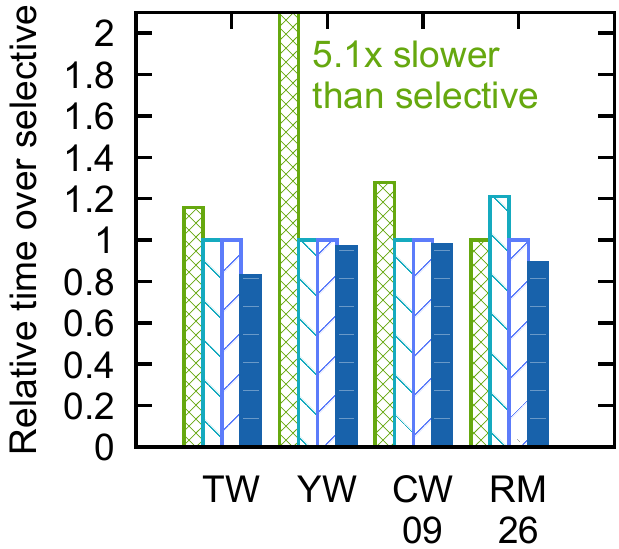}
    \caption{Running time}
 		\label{fig:block-placement-runtime}
  \end{subfigure}
  \hfil
  \begin{subfigure}[t]{0.49\linewidth}
    \includegraphics[width=40mm]{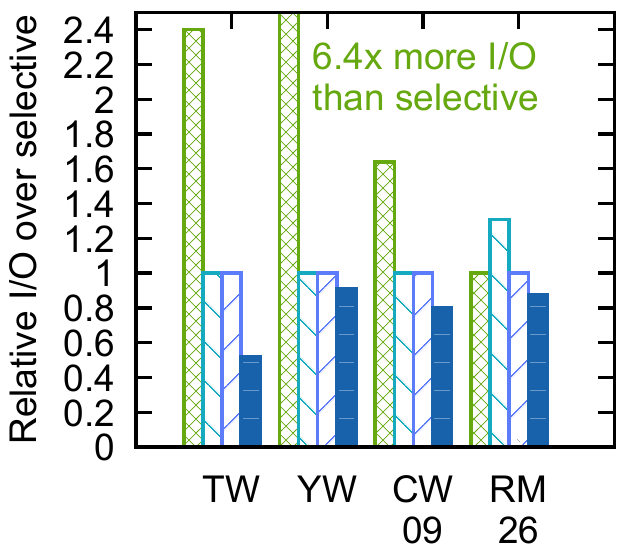}
    \caption{Amount of I/O}
		\label{fig:block-placement-io}
  \end{subfigure}
	\vspace{-3mm}
	\caption{The effect of the matrix density on running time and I/O.
		\methodV is faster and more I/O efficient than \methodH for sparse graphs while \methodH is faster and more I/O efficient than \methodV for a dense graph.
		\methodHB shows the best performance for all cases outperforming other versions of \method.
	}
	\label{fig:block-placement}
\end{figure}

We perform experiments to answer the following questions:

\begin{itemize}
  \item[Q1.] How much does \method improve the performance and scalability compared to the existing systems? (Section~\ref{sec:experiment-performance})
  \item[Q2.] How much does the matrix density affect the performance of the \method's four methods? (Section~\ref{sec:experiment-performance-density})
  \item[Q3.] How much does the threshold $\theta$ affect the performance and the amount of I/O of \methodHB? (Section~\ref{sec:experiment-threshold})
  \item[Q4.] How does \method scale up with the number of workers? (Section~\ref{sec:experiment-machine-scalability})
  \item[Q5.] How does the performance of \method differ depending on the underlying distributed framework? (Section~\ref{sec:experiment-underlying-engine})
\end{itemize}

\begin{table}[!t]
  \centering
  \small
  \setlength{\tabcolsep}{1.5pt}
  \caption{\textbf{The summary of graphs.}}
  \vspace{-4mm}
  \begin{center}
    \begin{tabularx}{0.44\textwidth}{rrrr}
      \toprule
      \textbf{Graph} & \textbf{Vertices} & \textbf{Edges} & \textbf{Source} \\
      \midrule
      ClueWeb12 (CW12) & 6,231,126,594 & 71,746,553,402 & Lemur Project\footnotemark[1] \\
      ClueWeb09 (CW09) & 1,684,876,525 & 7,939,647,897 & Lemur Project\footnotemark[2] \\
      YahooWeb (YW) & 720,242,173 & 6,636,600,779 & Webscope\footnotemark[3] \\
      Twitter (TW) & 41,652,230 & 1,468,365,182 & Kwak et al.\footnotemark[4]~\cite{Kwak10www} \\
      \midrule
      RMAT26 (RM26) & 42,147,725 & 5,000,000,000 & TegViz.\footnotemark[5]~\cite{DBLP:conf/icdm/JeonJK15} \\
      \bottomrule
    \end{tabularx}
  \end{center}
  \label{tbl:datasets}
  \vspace{-2mm}
\end{table}

\footnotetext[1]{\url{https://lemurproject.org/clueweb12/}}
\footnotetext[2]{\url{https://lemurproject.org/clueweb09/}}
\footnotetext[3]{\url{http://webscope.sandbox.yahoo.com}}
\footnotetext[4]{\url{http://an.kaist.ac.kr/traces/WWW2010.html}}
\footnotetext[5]{\url{http://datalab.snu.ac.kr/tegviz}}

\subsection{Datasets}
\label{sec:datasets}
\vspace{-1mm}

We use real-world graphs to compare \method to existing systems (Sections~\ref{sec:experiment-performance} and \ref{sec:experiment-machine-scalability}) and a synthetic graph to evaluate the performance of \method (Section~\ref{sec:experiment-performance-density}).
The graphs are summarized in Table~\ref{tbl:datasets}.
\textit{Twitter} is a who-follows-whom network in Twitter crawled in 2010.
\textit{YahooWeb}, \textit{ClueWeb09} and \textit{ClueWeb12} are page-level hyperlink networks on the WWW. 
RMAT~\cite{DBLP:conf/sigcomm/FaloutsosFF99} is a famous graph generation model that matches the characteristic of real-world networks.
We generate an RMAT graph with parameters $a=0.57$, $b=0.19$, $c=0.19$, and $d=0.05$ using TegViz~\cite{DBLP:conf/icdm/JeonJK15}, a distributed graph generator.



\subsection{Environment}
\label{sec:experiment-environment}
\vspace{-1mm}

We implemented \method on Hadoop and Spark, which are famous distributed processing frameworks.
Sections~\ref{sec:experiment-performance}, \ref{sec:experiment-performance-density}, and \ref{sec:experiment-machine-scalability} show the experimental results on Hadoop.
The result on Spark is in Section~\ref{sec:experiment-underlying-engine}.
%
We compare \method to existing graph processing systems: PEGASUS, GraphX, GraphLab, and Giraph. 
PEGASUS is a disk-based system, and the others are distributed-memory based systems.

We run our experiments on a cluster of 17 machines; one is a master and the others are for workers.
Each machine is equipped with an Intel E3-1240v5 CPU (quad-core, 3.5GHz), 32GB of RAM, and 4 hard disk drives.
A machine that is not the master runs 4 workers, each with 1 CPU core and 6GB of RAM.
All the machines are connected via 1 Gigabit Ethernet.
Hadoop 2.7.3, Spark 2.0.1 and MPICH 3.0.4 are installed on the cluster.

\subsection{Performance of \method}
\label{sec:experiment-performance}
\vspace{-1mm}

We compare the running time of \method and competitors (PEGASUS, GraphX, GraphLab, and Giraph) on ClueWeb12; induced subgraphs with varying number of edges are used.
For each system, we run the PageRank algorithm with 8 iterations.
Figure~\ref{fig:data-scale-runtime} shows the running time of all systems on various graph sizes.
We emphasize that only \method succeeds in processing the entire ClueWeb12 graph.
The memory-based systems fail on graphs with more than 2.3 billion edges due to out of memory error,
while PEGASUS fails to process graphs with 9 billion edges within 5 hours.
The underlying causes are as follows.
Giraph requires that all the out-edges of the assigned vertices are stored in the main memory of the worker.
However, this requirement can be easily broken since highly skewed degree distribution is likely to lead to out of memory error.
GraphLab uses the vertex-cut partitioning method and copies the vertices to the multiple workers which have the edges related to the vertices.
The edges and the copied vertices are stored in the main memory of each worker, and incur the out of memory error.
GraphX uses the same approach as GraphLab, but succeeds in processing a graph which GraphLab fails to process because Spark, its underlying data processing engine, uses both the disk and the main memory of each worker.
Even GraphX, however, fails to process graphs with more than 2.3 billion edges due to huge number of RDD partitions.

\subsection{Effect of Matrix Density}
\label{sec:experiment-performance-density}
\vspace{-1mm}

We evaluate the performance 
of \method on graphs with varying density.
The results are in Figure~\ref{fig:block-placement}.
Twitter, YahooWeb, and ClueWeb09 are real-world sparse graphs where the matrix density $|M|/|v|^2$ is less than $10^{-7}$ while RMAT26 is a synthetic dense graph where the matrix density is larger than $10^{-7}$.
As we discussed in Section~\ref{subsec:pmv-selective}, the vertical placement is appropriate for a sparse graph while the horizontal placement is appropriate for a dense graph.
Figures \ref{fig:block-placement-runtime} and \ref{fig:block-placement-io} verify the relation between the performance and the density of graph.
\methodV shows a better performance than \methodH when the input matrix is sparse.
On the other hand, if the matrix is dense, \methodH provides a better performance than \methodV.
\methodS shows the same performance as the best of \methodH and \methodV as we expected.
\methodHB significantly reduces the amount of I/O for both sparse and dense graphs, and improves the performance up to 18\% from \methodS.

\subsection{Effect of Threshold $\theta$ }
\label{sec:experiment-threshold}
\vspace{-1mm}

\begin{figure}[!tp]
  \centering
  \includegraphics[width=80mm]{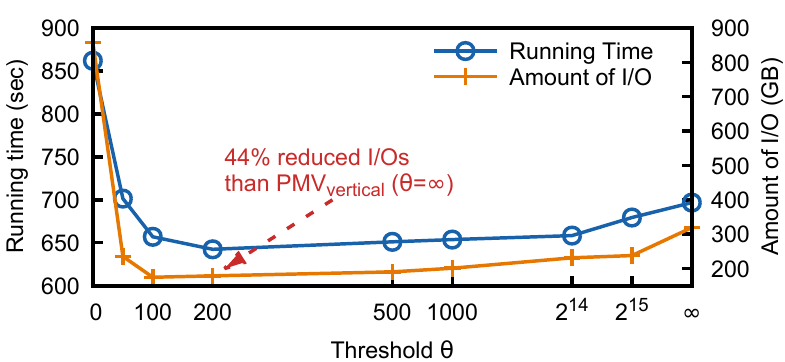}
  \vspace{-3mm}
  \caption{The effect of threshold $\theta$ on the running time and the amount of I/O. \methodHB shows the best performance when $\theta = 200$ with 44\% reduced amount of I/O compared to when $\theta = \infty$, i.e., \methodV.
  }
  \label{fig:threshold}
  \vspace{-2mm}
\end{figure}

We iterate \methodHB based PageRank algorithm 30 times on Twitter graph varying threshold $\theta$.
Figure~\ref{fig:threshold} presents the effect of the threshold on the running time and the amount of I/O.
\methodV ($\theta = \infty$) shows better performance and lower amount of I/O than \methodH ($\theta = 0$), as we expected, because Twitter is sparse with density lower than $10^{-7}$.
\methodHB achieves the best performance with $\theta = 200$: in the setting
\methodHB shows 44\% decreased amount of I/O compared to that of \methodV, from 318GB to 178GB.
{\color{blue}
}
Note that $\theta=100$ gives the minimum amount of I/O while $\theta=200$ gives the fastest running time.
A possible explanation is that skewness of in-degree distribution of dense area and out-degree distribution of sparse area affects the running times of horizontal and vertical computations of \methodHB, respectively; however, the difference is minor and does not change the conclusion that \methodHB outperforms all other versions of \method.
%

\subsection{Machine Scalability}
\label{sec:experiment-machine-scalability}
\vspace{-1mm}

\begin{figure}[!t]
	\centering
	\includegraphics[width=85mm]{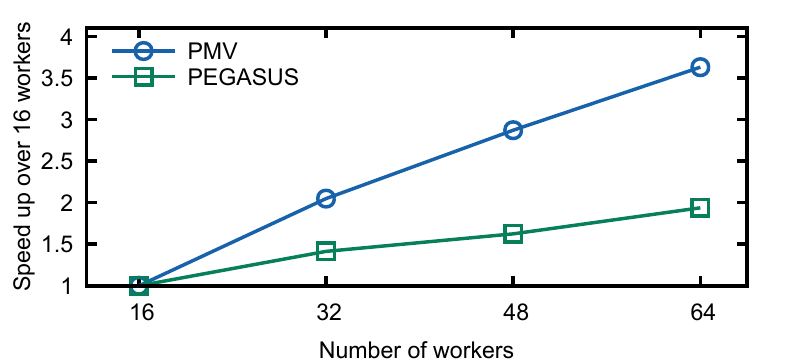}
	\vspace{-3mm}
	\caption{Machine scalability of \method on YahooWeb. \method shows linear machine scalability with slope close to 1, while PEGASUS does with a much smaller slope because of the curse of the last reducer problem~\cite{DBLP:conf/www/SuriV11} incurred by the high-degree vertices.}
	\label{fig:machine-scale-runtime}
\end{figure}


We evaluate the machine scalability of \method and competitors by running the PageRank algorithm with varying number of workers on YahooWeb.
Figure~\ref{fig:machine-scale-runtime} shows the speedup according to the number of workers from 16 to 64;
the speedup is defined as $t_{16}/t_n$, where $t_n$ is the running time with $n$ workers.
%
We omit GraphLab, Giraph, and GraphX because they fail to process the YahooWeb graph on 16 workers.
\method shows linear machine scalability with slope close to 1, while PEGASUS does with a much smaller slope.
PEGASUS suffers from the curse of the last reducer problem~\cite{DBLP:conf/www/SuriV11} which is incurred by the high-degree vertices.
\method overcomes the problem by treating the high-degree vertices in multiple workers.


\subsection{Underlying Engine}
\label{sec:experiment-underlying-engine}
\vspace{-1mm}

\begin{figure}[!t]
  \centering
  \includegraphics[width=87mm]{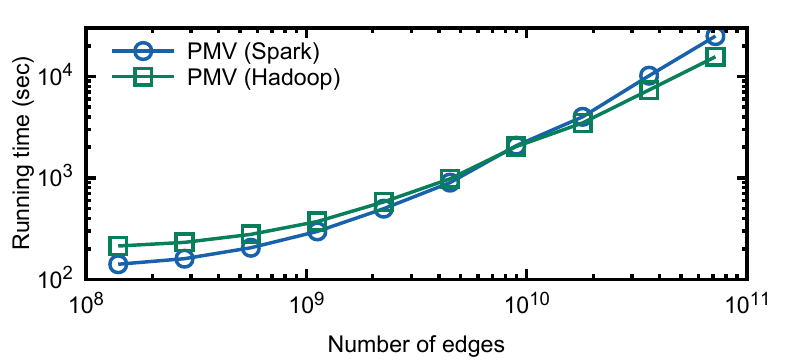}
  \vspace{-3mm}
  \caption{The running time of \method on Hadoop and Spark.
    PMV is faster on Spark than on Hadoop when the graph is small. On large graphs, however, PMV runs faster on Hadoop than on Spark (see Section~\ref{sec:experiment-underlying-engine} for details).
  }
  \label{fig:spark-real-runtime}
\end{figure}

Figure~\ref{fig:spark-real-runtime} shows the performance of PMV according to underlying systems: Hadoop and Spark.
We use ClueWeb12 with varying number of edges as in Section~\ref{sec:experiment-performance}.
When the graph is small, PMV on Spark beats PMV on Hadoop.
This is because Spark is highly optimized for iterative computation;
Spark requires much less start-up and clean-up time for each iteration than Hadoop does.
When the graph is large, however, PMV on Spark falls behind PMV on Hadoop.
PMV on Spark requires more memory than PMV on Hadoop since Spark's RDD is immutable;
for updating a vector, PMV on Spark creates a new vector requiring additional memory while PMV on Hadoop updates the vector in-place.
Accordingly,
%
when the graph is large, \method on Spark needs to partition the input vector into smaller blocks than \method on Hadoop does.
This makes the performance of \method on Spark worse than that of \method on Hadoop for large graphs.

\section{Conclusion}
\label{sec:conclusion}
\vspace{-1mm}
We propose \method, a scalable graph mining method based on generalized matrix-vector multiplication on distributed systems.
\method exploits both horizontal and vertical placement strategies
to reduce I/O costs.
\method shows up to $16 \times$ larger scalability than existing distributed memory methods, $9 \times$ faster performance than existing disk-based ones, and linear scalability for the number of edges and machines.
Future research directions include a graph partitioning algorithm that improves the performance of graph mining algorithms based on distributed matrix-vector multiplication.

\bibliographystyle{ACM-Reference-Format}
\bibliography{references}

\appendix

\end{document}